\documentclass[letterpaper, 10 pt, conference]{ieeeconf}  
\IEEEoverridecommandlockouts                 
\usepackage{graphicx}
\usepackage{amsmath,amssymb}
\usepackage{ntheorem}
\usepackage{mathrsfs}  
\usepackage{enumerate}

\newtheorem{assum}{Assumption}

\newtheorem{theorem}{Theorem}
\newtheorem{lemma}{Lemma}
\newtheorem{property}{Property} 

\usepackage[ruled,linesnumbered]{algorithm2e}
\usepackage{algpseudocode}
\usepackage{cite} 
\usepackage[dvipsnames, table]{xcolor} 

\newcommand{\conn}{\lambda_2}
\newcommand{\vect}[1]{\boldsymbol{#1}}
\newcommand{\cbf}{\conn(x) - \epsilon}
\DeclareMathOperator*{\argmin}{argmin}

\usepackage{nicefrac}
\usepackage{paralist}
\usepackage[normalem]{ulem}
\usepackage{caption}
\usepackage{multirow}  
\usepackage{booktabs}  
\usepackage{array}
\usepackage{diagbox}  
\usepackage{siunitx}

\usepackage{tikz}
\usepackage{pgfplots}	
\usepgfplotslibrary{groupplots}
\usepackage{siunitx}
\usepackage{pgfplotstable}
\usepackage{caption}
\usepackage[labelformat=simple]{subcaption} 
\usepackage{pifont} 

\DeclareCaptionType{Table}
 
\graphicspath{{images/}}
 
\definecolor{color1}{rgb}{0,0.4470,0.7410}
\definecolor{color2}{rgb}{0.8500,0.3250,0.0980}
\definecolor{color3}{rgb}{0.9290,0.6940,0.1250}
\definecolor{color4}{rgb}{0.4940,0.1840,0.5560}
\definecolor{color5}{rgb}{0.4660,0.6740,0.1880}
\definecolor{color6}{RGB}{19,122,16}

\pgfplotscreateplotcyclelist{MyCyclelist}{%
	{color1, mark = none, line width=0.6pt},
	{color3, mark = none, line width=0.6pt},
	{color4, mark = none, line width=0.6pt},
	{color5, mark = none, line width=0.6pt},
	{color4, mark = none, line width=0.6pt}
}

\definecolor{my_orange}{RGB}{255,121,3}
\definecolor{my_light_green}{RGB}{32,255,42}
\definecolor{my_yellow}{RGB}{255,255,32}
\definecolor{my_light_blue}{RGB}{50,255,255}
\definecolor{my_blue}{RGB}{50,120,255}

\begin{document}
\title{\LARGE \bf
	Connectivity Maintenance: Global and Optimized approach through Control Barrier Functions }

\author{Beatrice Capelli, Lorenzo Sabattini%
	%
	%
	%
	\thanks{Authors are with the Department of Sciences and Methods for Engineering (DISMI), University of Modena and Reggio Emilia, Italy {\tt\small{\{beatrice.capelli, lorenzo.sabattini\}@unimore.it}}}%
}

%

\maketitle

\begin{abstract}                          
Connectivity maintenance is an essential aspect to consider while controlling a multi-robot system. In general, a multi-robot system should be connected to obtain a certain common objective. Connectivity must be kept regardless of the control strategy or the objective of the multi-robot system. Two main methods exist for connectivity maintenance: keep the initial connections (local connectivity) or allow modifications to the initial connections, but always keeping the overall system connected (global connectivity). 
In this paper we present a method that allows, at the same time, to maintain global connectivity and to implement the desired control strategy (e.g., consensus, formation control, coverage), all in an optimized fashion. For this purpose, we defined and implemented a Control Barrier Function that can incorporate constraints and objectives. We provide a mathematical proof of the method, and we demonstrate its versatility with simulations of different applications.
\end{abstract}


\section{Introduction}
Multi-robot systems are becoming more and more frequent in a wide range of domains, such as industrial~\cite{ram2017}, agricultural~\cite{ball2015}, marine~\cite{hajieghrary2016,kemna2018}, and aerial~\cite{saska2017}.
%
%
Communication among the robots is crucial for most cooperative applications, such as search and rescue~\cite{hayat2017multi}, patrolling~\cite{portugal2011survey, pasqualetti2012cooperative}, or exploration~\cite{cortes2010coverage, burgard2005coordinated}.
%
%
Hence, the controller for a multi-robot system should take into account communication, which is usually represented by a communication graph, and should be able to maintain connectivity (i.e., possibility for the robots to exchange information among each other).

However, in several applications of interest, the desired input, namely the input defined to achieve the desired objective, 
%
can lead to a disconnection of the group. For example, if we want to perform measurements 
with sensors and maximize the area we analyze, we should control the robots in such a way that they spread as much as possible all over the area. But, as they usually have a limited range of communication, we may cause the disconnection of the group. 

In general, to maintain connectivity there are two approaches: local and global. The local approach consists in keeping the connection between two robots if it exists when the task begins, so if the group starts connected it will be connected during the overall task. The main disadvantage of this method is the low flexibility due to the inability to change the communication graph. However, it can be easily implemented in a decentralized manner, thanks to the fact that it only needs local information, and in literature there exist a lot of examples~\cite{ji2007distributed, dimarogonas2008decentralized, ajorlou2010class}. On the other side, global connectivity consists in keeping the group connected allowing a disconnection between the robots, if such disconnection does not lead to a loss of communication in the group. This method allows a higher flexibility, allowing the robots to re-arrange their interconnections to accommodate external needs, such as the presence of obstacles that obstruct their motion (e.g., a narrow passage). 
In~\cite{yang2010decentralized, sabattini2013decentralized, li2013bounded} a decentralized connectivity maintenance is built through a gradient descent approach, while a recent comparison between the two approaches can be found in~\cite{khateri2019comparison}. In order to increase the flexibility and the efficiency of the multi-robot system, we aim at considering global connectivity preservation together with some desired control input, and combine them in such a way to optimize the task while maintaining communication constraints.

%
%

Control Barrier Functions (CBFs) are a suitable method for our problem. In fact, they allow to consider, at the same time, different objectives and constraints. CBFs are commonly used for safety critical applications, such as in the automotive field where they are used for adaptive cruise control~\cite{ames2014control} and line keeping~\cite{xu2017realizing}. Another field of application is human-robot interaction, where CBFs allow to respect some safety constraints and, at the same time, optimize the task~\cite{landi2019safety}. In the multi-robot domain, CBFs are used in~\cite{wang2017safety} to guarantee collision-free movement of a group of mobile robots. In~\cite{egerstedt2018robot} they are exploited to combine different aspects: safety, connectivity, coverage and energy management. Moreover, in~\cite{wenhao2019} CBFs are used to optimize connectivity and coverage. In~\cite{egerstedt2018robot, wenhao2019} the approach towards connectivity is local because the CBFs focus on maintaining all the initial links between the robots. 


In this paper we propose to integrate global connectivity into a desired control (e.g., coverage, formation control, patrolling), by means of a CBF. This method allows to optimize the desired behavior of the system and, at the same time, to maintain connectivity. The optimization process allows to avoid the tuning procedure,  which is usually present when two or more objectives or constraints are considered in the same controller~\cite{siligardi2019robust}. Like in other global approaches, the initial configuration should generate a connected graph. But, differently from others (e.g.,~\cite{sabattini2013decentralized}), where the initial value of connectivity should be above a certain desired value, we are able to increase connectivity to the desired value, or above. This property is due to the fact that the chosen CBF is also a Control Lyapunov Function (CLF). 


In the paper we give a general introduction of connectivity and of CBFs, we give the demonstration of the implemented method and an extensive evaluation through simulations. 
\section{Notation}
The symbol $\mathcal{L}$ will represent the set of locally Lipschitz functions. $\mathbb{R}$, $\mathbb{R}^+_0$ are the set of real and real non-negative numbers, respectively. A continuous function \mbox{$\Omega(\cdot): \mathbb{R}^+_0 \rightarrow \mathbb{R}^+_0$} is a class $\mathcal{K}$ function if it is strictly increasing and $\Omega(0) = 0$. It is an extended class $\mathcal{K}$ if it is a class $\mathcal{K}$ function, and it is defined on the entire real line, e.g., $\Omega(\cdot): \mathbb{R} \rightarrow \mathbb{R}$.

\section{Background}


\subsection{Connectivity}
\label{sec:conn}
We consider a group of robots, which have a limited communication range $R$. The communication network is represented by a communication graph\footnote{All the formulation is time dependent because the communication network can change over time (e.g., $\mathcal{G} = \mathcal{G}(t)$). However, when not strictly necessary, we omit dependence on time, for ease of notation.} $\mathcal{G}=\{ V, E\}$, where the set of vertices $V$ represents the robots and the set of the edges $E$ consists of the existing communication links. An edge $e_{i,j}$ exists if and only if the $i$-th and the $j$-th robot are within the communication range $R$. The existence of the edge $e_{i,j}$ implies also the existence of the edge $e_{j,i}$, because we are considering an undirected graph. We define the set of neighbors of the $i$-th robot as $\mathcal{N}_i = \{ j \in V |  e_{i,j} \in E \}$, which are all the robots that can communicate with it. 

Algebraic connectivity of the system is available directly from the mathematical representation of the graph $\mathcal{G}$. Firstly, we introduce the adjacency matrix $A \in \mathbb{R}^{N \times N}$:
\begin{equation}
A = \begin{cases}
a_{i,j} > 0 & \text{if } j \in \mathcal{N}_i \\
0 & \text{otherwise}
\end{cases}
\end{equation}
where $a_{i,j}$ is the edge weight of $e_{i,j}$. 

Then we define the degree matrix $D \in \mathbb{R}^{N \times N}$, which is a diagonal matrix and each element of the diagonal is equal to $\psi_{i,i} = \sum_{j=1}^{N}a_{i,j}$. Hence, $D = diag\{\psi_{i,i}\}$.

Combining the adjacency and the degree matrix, we obtain the (weighted) Laplacian matrix $L = D - A$. The Laplacian matrix has some important properties~\cite{mohar1991laplacian}:
\begin{itemize}
	\item $L \vect{1} = 0$, which means that the first eigenvalue of the Laplacian is zero and the associated eigenvector is $\vect{1}$, namely the column vector of all ones;
	\item Considering the set of eigenvalues of the Laplacian matrix ($\lambda_1, \dotsc, \lambda_N $) we know that:
	\begin{itemize}
		\item  They can be ordered as: $0= \lambda_1 \leq \lambda_2 \leq \dotsc \leq \lambda_N$;
		\item The second eigenvalue $\lambda_2$ represents the algebraic connectivity of the graph, and we have $\lambda_2 > 0$ if and only if $\mathcal{G}$ is connected;
	\end{itemize}
\end{itemize}

\subsection{Control Barrier Functions}
\label{sec:cbf}
In this section we introduce the concept of Control Barrier Functions (CBFs), which are an extension of Barrier Functions for systems that do not present a forward invariant set for the desired state of the system. For more details the reader is referred to~\cite{ames2016control}.

We consider an affine control system:
\begin{equation}
\label{eq:affine}
\dot{\chi} = f(\chi) + g(\chi)\mu
\end{equation}
where $\chi \in \mathbb{R}^p$ represents the state of the system, \mbox{$\mu \in U \subseteq \mathbb{R}^q$} is the control input, $f(\chi), g(\chi) \in \mathcal{L}$.

What CBFs aim to do is keeping, or eventually carrying, the system~\eqref{eq:affine} in a certain closed set $\mathcal{C} \subset \mathbb{R}^p$ defined as:
\begin{equation}
\label{eq:set}
\begin{aligned}
\mathcal{C} &= \{ \chi \in \mathbb{R}^p | h(\chi) \geq 0 \} \\
\partial \mathcal{C} &= \{ \chi \in \mathbb{R}^p | h(\chi) = 0 \} \\
\text{Int}(\mathcal{C}) &= \{ \chi \in \mathbb{R}^p | h(\chi) > 0 \}
\end{aligned}
\end{equation} 
where $h(\chi) : \mathbb{R}^p \rightarrow \mathbb{R}$ is a continuously differentiable function. At this point there are two different approaches to define the control strategy: Reciprocal Control Barrier Functions (RCBFs) and Zeroing Control Barrier Functions (ZCBFs). For RCBFs we must define a function $B(\chi)$ that tends to infinity while the state approaches the boundary of the set $\mathcal{C}$, while for ZCBFs the function $h(\chi)$ must tends to zero on the boundary. We decided to use ZCBFs because they do not present the drawbacks of a function that assumes unbounded values on the boundary and they allow also to consider some perturbations or errors in the system (defining the function $h(\chi)$ on a set $\mathcal{D}$ larger than $\mathcal{C}$, $\mathcal{C} \subseteq \mathcal{D} \subset \mathbb{R}^p$). In addition, under certain conditions, they can be straightforwardly transformed in Control Lyapunov Functions, because they are well defined also outside of the set $\mathcal{C}$. In the following we will use the term CBF meaning ZCBF. 

The CBF should keep the state of the system inside\footnote{Formally, to make the set $\mathcal{C}$ \textit{forward invariant} (i.e., if the system starts in whatever state $\chi(0) \in \mathcal{C}$, then $\chi(t) \in \mathcal{C}, \text{ } \forall t > 0$).} the set $\mathcal{C}$, and this is guaranteed choosing the input $\mu$ in~\eqref{eq:affine} in such a way that:
\begin{equation}
\label{eq:derivative}
\frac{d}{dt} h(\chi) \geq -\alpha\left(h(\chi)\right)
\end{equation} 
where $\alpha(\cdot) : \mathbb{R}^p \rightarrow \mathbb{R}^p$ is an extended class $\mathcal{K}$ function\footnote{Typically, $\alpha(h)$ is chosen equal to a linear function $\alpha(h(\chi))=k h(\chi)$ or to the cube $\alpha(h(\chi)) = k h(\chi)^3$, with $k>0$.} and $\alpha(\cdot) \in \mathcal{L}$. Using~\eqref{eq:affine}, we can rewrite~\eqref{eq:derivative} as an explicit function of $\mu$, as follow:
\begin{equation*}
\frac{d}{dt} h(\chi) = \frac{\partial h(\chi)}{\partial \chi}\dot{\chi} = L_f h(\chi) + L_g h(\chi)\mu \geq -\alpha\left(h(\chi)\right)
\end{equation*}
where $L_f$ and $L_g$ are the Lie derivatives of $h(\chi)$: \mbox{$L_f h(\chi) = \frac{\partial h(\chi)}{\partial \chi} f(\chi)$,} $L_g h(\chi) = \frac{\partial h(\chi)}{\partial \chi} g(\chi)$.

Then, the function $h(\chi)$ is a CBF if the following properties are satisfied~\cite{ames2016control}: 
\begin{property}\label{prop:contdiff}
	$h(\chi)$ is continuously differentiable.
\end{property}
\begin{property}\label{prop:reldeg}
	$h(\chi)$ is of relative degree one (i.e., its first order time derivative depends explicitly on the control input).
\end{property} 
\begin{property}\label{prop:existence_input}
	It is possible to find an extended class $\mathcal{K}$ function $\alpha\left(h(\chi)\right)$ such that:
	\begin{equation}
	\label{eq:sup}
	\sup\limits_{u \in U} \left[ L_f h(\chi) + L_g h(\chi)\mu + \alpha\left(h(\chi)\right) \right] \geq 0, \text{ } \forall \chi \in \mathcal{D}
	\end{equation} 
\end{property}
In order to make the set $\mathcal{C}$ forward invariant, we can apply any controller $\mu(\chi): \mathcal{D} \rightarrow U \text{, } \mu(\chi) \in \mathcal{L}$ such that \mbox{$\mu(\chi) \in K_{cbf}(\chi)$}. The set $K_{cbf}(\chi)$ is defined as:
\begin{equation*}
K_{cbf}(\chi) = \{ \mu \in U | L_f h(\chi) + L_g h(\chi)\mu + \alpha\left(h(\chi)\right) \geq 0 \}
\end{equation*}
It is important to notice that the closed loop controller $\mu(\chi) \in K_{cbf}(\chi)$ does not ensure that the system will converge to the given set $\mathcal{C}$ if it starts outside of it. 

For this purpose, we can introduce a Control Lyapunov Function (CLF), which is a continuously differentiable, positive definite function $\Phi(\chi): \mathcal{W} \subset \mathbb{R}^p \rightarrow \mathbb{R}^+_0$ such that:
\begin{equation*}
\inf\limits_{\mu \in U} \left[ L_f \Phi(\chi) + L_g\Phi(\chi)\mu \right] \leq -\zeta\left(\Phi(\chi)\right), \text{ } \forall \chi \in \mathcal{W}
\end{equation*}
where $\zeta$ is a class $\mathcal{K}$ function. Hence, it is possible to find an input $\mu \in U$ that stabilizes a point $\chi* \in \mathbb{R}^p$, or a set. More details can be found in~\cite{khalil2002nonlinear}. 

\section{System definition and problem statement}\label{sec:prob}
We consider a system composed by $N$ robots moving in a $n$-dimensional space, and we define the position of the $i$-th robot as $x_i \in \mathbb{R}^n$. Then we can stack all the positions in a vector, which will represent the state of the system: $x = \left[ x_1^T, \dotsc, x_N^T \right]^T \in \mathbb{R}^{nN}$.
%
%
We consider single-integrator dynamics:
\begin{equation}
\label{eq:sinint}
\dot{x} = u
\end{equation}
where $u \in U \subseteq \mathbb{R}^{nN}$ is the control input of the system. Hence,~\eqref{eq:sinint} can be represented by the general model of an affine control system, reported in~\eqref{eq:affine}, with
%
\mbox{$f(x) = \mathbb{O} \in \mathbb{R}^{nN \times nN}$}, $g(x) = \mathbb{I} \in \mathbb{R}^{nN \times nN}$,
where $\mathbb{O}$ and $\mathbb{I}$ are, respectively, the null and the identity matrix.

As in~\cite{tro2017}, we define the edge weights of the graph $\mathcal{G}$ as: 
\begin{equation}
\label{eq:weight}
a_{i,j} = \begin{cases}
e^{\left(R^2 - d_{ij}^2\right)^2/\sigma} -1 & \text{if\ } d_{i,j} \leq R \\
0 & \text{otherwise}
\end{cases}
\end{equation}
where $d_{i,j} = \| x_i - x_j \|$ is the Euclidean distance between the $i\text{-th}$ and the $j$-th robot, $R >0$ is the communication distance, and $\sigma>0$ is a positive constant to set the edge weight $a_{i,j} \leq 1$, for normalization purpose. The above definition of the weight respects the following constraints:
\begin{itemize}
	\item $a_{i,j} \geq 0$ to have similar properties to the unweighted Laplacian~\cite{mohar1991laplacian}, and $a_{i,j} = 0$ if $d_{i,j} > R$;
	\item continuously differentiable with respect to the distance between the robots;
	\item decreasing as the distance increments.
\end{itemize}

We will hereafter assume that the Laplacian matrix has simple eigenvalues, i.e., $\lambda_i\neq\lambda_{i+1}, \forall i \in \left[2, \dotsc, N-1 \right]$.

This assumption is almost always verified in the considered scenario, in which edge weights are defined, in~\eqref{eq:weight}, as a function of the time-varying inter-robot distance. In fact, as shown in~\cite{poignard2018spectra}, in the pathological case in which the Laplacian has multiple eigenvalues, a small perturbation (introduced, in our case, by the motion of the robots) in the edge weights is sufficient to recover the simple eigenvalues situation.

%

In this paper, we aim at solving the following problem: a multi-robot system, starting from an initial configuration where the group is connected ($\conn(x(0)) > 0$), remains connected, while being controlled to achieve some desired objective. 

From a mathematical point of view, we want to minimize the difference between the desired control input \mbox{$u_{des} \in U \subseteq \mathbb{R}^{nN}$} and the actual input $u \in U \subseteq  \mathbb{R}^{nN}$, which should ensure the connectivity of the system, encoded by $h(x)$. 
This problem will be solved, exploiting the CBF method, by means of a Quadratic Program (QP) subject to constraints:
\begin{equation}
\label{eq:qp}
\begin{aligned}
u(x) = \argmin\limits_{u \in \mathbb{R}^{nN}} \quad & \frac{1}{2}\| u - u_{des}(x) \|^2\\
\text{s.t. } \quad & L_f h(x) + L_g h(x) u \geq - \alpha \left(h(x)\right) \\
& u(x) \in \mathcal{L}
\end{aligned}
\end{equation}

\section{CBF for connectivity maintenance}

In order to solve the problem stated in Section~\ref{sec:prob}, we propose the following candidate CBF:
\begin{equation}
\label{eq:cbf}
h(x) = \cbf
\end{equation}
where $\conn(x)$ is the second eigenvalue of the Laplacian matrix 
and $\epsilon > 0$ is introduced in order to keep connectivity above a desired value. In the following, we will omit the dependence of $\conn$ from the state $x$, for ease of notation.

The CBF $h(x)$ given in~\eqref{eq:cbf} defines the set $\mathcal{C}$, described as in~\eqref{eq:set}, in which we want to keep our system:
\begin{equation*}
\mathcal{C}= \{x \in \mathbb{R}^{nN} | \conn \ge \epsilon \}
\end{equation*}
In addition, the CBF~\eqref{eq:cbf} is defined on a set $\mathcal{D}$, with $\mathcal{C} \subset \mathcal{D}$:
\begin{equation*}
\mathcal{D} = \{ x \in \mathbb{R}^{nN} | \conn > 0 \}  
\end{equation*} 
This means that, in the following, we assume that the system always starts in a connected configuration.


We will now demonstrate the existence of a solution to the QP problem~\eqref{eq:qp}, choosing $\alpha\left(h(x)\right) = \varphi \cdot h(x)$, with $\varphi > 0$.
%
Without loss of generality, in the following we will consider $\varphi = 1$, for ease of notation. Then we obtain:
\begin{equation}
\label{eq:full_constraint}
\frac{\partial \conn}{\partial x} u + \left( \conn - \epsilon \right) \ge 0 
\end{equation}
We rename the partial derivative with $\beta$:
\begin{equation}
\label{eq:beta}
\beta = \left[\beta_1, \ldots, \beta_N\right] = \frac{\partial \conn}{\partial x} = \left[ \frac{\partial \conn}{\partial x_1}, \dotsc, \frac{\partial \conn}{\partial x_N}  \right]
\end{equation}

As shown in~\cite{yang2010decentralized}, the $i$-th component of $\beta$ can be written as:
%
\begin{equation}
\label{eq:der_conn}
\beta_i = \frac{\partial \conn}{\partial x_i}
 = \sum_{j \in \mathcal{N}_i}  \frac{\partial a_{i,j}}{\partial x_i} \left(v_2^i - v_2^j\right)^2  
\end{equation}
where $v_2^i$ and $v_2^j$ are the $i$-th and $j$-th component of the eigenvector associated to $\conn$, respectively.

We will now show that at least one component of $\beta$ is different from zero. This result is instrumental for the main results, that will be subsequently presented. We make the following assumption:
\begin{assum}\label{ass:collav}
	A collision avoidance mechanism is implemented, that prevents robots from colliding among each other.
\end{assum}
\begin{lemma}
	\label{lemma:beta_k} Consider a multi-robot system with communication topology defined according to the edge weights given in~\eqref{eq:weight}, and consider the definition of $\beta$ given in~\eqref{eq:beta}. Let the communication graph be connected. Then, a value \mbox{$k \in [1, \dotsc, N]$} always exists such that $\beta_k \neq 0$.
\end{lemma}

\begin{proof}
	For ease of notation the demonstration is carried on component-wise, considering $x_i \in \mathbb{R}$. All the results can be trivially extended to the multi-dimensional case.
	
	To prove the statement we first divide the communication graph $\mathcal{G}$ in $z$ subgraphs, where $z \geq 2$. The rule to divide the nodes in the subgraphs is their corresponding component inside the second eigenvector of the Laplacian matrix, namely $v_2$, which is the eigenvector associated to $\conn$. As discussed in Section~\ref{sec:conn}, the eigenvector associated to the null eigenvalue of the Laplacian matrix is the column vector of all ones, i.e., $\vect{1}$. Since the eigenvectors of a matrix are orthogonal, we know in particular that $v_2$ is orthogonal to $\vect{1}$, that is: \mbox{$v_2 \neq \rho \vect{1}$}, where \mbox{$\rho \in \mathbb{R} \setminus \{0\}$}. Then, defining $\gamma_1, \dotsc, \gamma_N \in \mathbb{R}$ as the components of $v_2$ (e.g., $\gamma_1 = v_2^{k}$ and $\gamma_2 = v_2^{j}$), we can always find at least two values such that $\gamma_1 \neq \gamma_2$. 
	
	We divide the communication graph $\mathcal{G}$ in $z$ subgraphs $\mathcal{G}_i$, where $i = 1, \dotsc, z$. Let:
	\begin{equation*}
	\begin{array}{lll}
	\mathcal{G}_1 = \{V_1, E_1\} & \dotsc & \mathcal{G}_z = \{V_z, E_z\} \\
	V_1 = \{ \omega \in V | v_2^{\omega} = \gamma_1 \} & \dotsc & V_z = \{ \omega \in V | v_2^{\omega} = \gamma_z\} \\
	E_1 = \{ e_{i,j} \in E | i,j \in V_1 \} & \dotsc & E_z = \{ e_{i,j} \in E | i,j \in V_z \}
	\end{array}
	\end{equation*}
	Hence, \mbox{$\mathcal{G} = \bigcup_{i=1}^{z} \mathcal{G}_{i} \cup \overline{\mathcal{G}}$}, where $\overline{\mathcal{G}}= \{\overline{V}, \overline{E}  \}$, \mbox{$V = \bigcup_{i=1}^{z} V_{i}$} and \mbox{$E = \bigcup_{i=1}^{z} E_{i} \cup \overline{E}$} 
	. 
	$\overline{E}$ represents the set of edges that link the subgraphs, namely the edges whose adjacent vertices belong to different subgraphs (the set of these vertices is $\overline{V}$):
	\begin{equation*}
	\overline{E} = \{e_{i,j} \in E | (i \in V_s \land  j \in V_t) \lor (i \in V_t \land  j \in V_s)  \} 
	\end{equation*} 
	with $ s, t \in Z$, $Z= \left[1, \dotsc, z \right]$.
	
	We define $i \in V$ an extreme node if:
	\begin{equation*}
	x_i < x_j \lor 	x_i > x_j, \forall j \neq i \in V
	\end{equation*}
	We know that two nodes always exist that verify this definition for~Assumption~\ref{ass:collav}.
	
	Consider now the subgraph $\overline{\mathcal{G}}$ and in particular one extreme node of it, namely the $k$-th vertex. For this vertex we demonstrate that $\beta_k \neq 0$. The term $\beta_k$ consists of a summation of terms, defined as the multiplication of two terms: $(v_2^k - v_2^j)^2$, and $\frac{\partial a_{k,j}}{\partial x_k}$. Considering the first term, namely $(v_2^k - v_2^j)^2$, we know that it is different from zero for all $j \neq k \in \overline{V}$, for the definition of the subgraphs. The second term $\frac{\partial a_{k,j}}{\partial x_k}$ can be divided in two additional terms $\left(\frac{\partial a_{k,j}}{\partial x_k} = \frac{\partial a_{k,j}}{\partial d_{k,j}}\frac{\partial d_{k,j}}{\partial x_k}\right)$. For $j \neq k \in \overline{V} \land j \in \mathcal{N}_k$, we can analyze the additional terms as follows, considering the definition of $a_{k,j}$ in~\eqref{eq:weight}:
	\begin{itemize}
		\item the term $\frac{\partial a_{k,j}}{\partial d_{k,j}}$ is different from zero, since $j\in\mathcal{N}_k$;
		\item the term $ \frac{\partial d_{k,j}}{\partial x_k}$ is different from zero if $d_{k,j}$ is different from zero, namely if the positions of robots $i$ and $j$ do not coincide. This is ensured according to Assumption~\ref{ass:collav}.
	\end{itemize}
	If the $k$-th vertex has just one neighbor in $\overline{\mathcal{G}}$, then the proof is completed. 
	
	If instead it has multiple neighbors in $\overline{\mathcal{G}}$ we need to consider the summation of multiple terms. In particular, we know that, if the summation is composed by terms of equal sign, then the summation will be different from zero. Considering the elements that build each addend, we know that the sign is defined by the difference between the position of the two vertices. 
	As $x_k$ is an extreme vertex then it will have all these terms with the same sign, which ends the demonstration.

\end{proof}

If we can ensure that the optimization problem~\eqref{eq:qp} admits always a solution, then we can guarantee that the multi-robot system stays connected during the desired task. 
 
\begin{theorem}
	\label{theorem:conn_constraint}
	Given the CBF described in~\eqref{eq:cbf}, consider the QP problem described in~\eqref{eq:qp}. Then, we can always find $u \in U$ that verifies the constraint~\eqref{eq:full_constraint} of the QP problem. 
\end{theorem}
\begin{proof}
	Substituting $\beta$, defined in~\eqref{eq:beta}, in the constraint of the QP problem~\eqref{eq:full_constraint}, we obtain $\beta u \geq \epsilon - \conn$. 
	We note that, if we can find an input $u$ such that: 
	\begin{equation}
	\label{eq:reduced_constraint}
	\beta u \geq \epsilon
	\end{equation}
	then this input will also satisfy the main constraint~\eqref{eq:full_constraint} because, from the characteristics of the Laplacian matrix, we know that $\conn \geq 0$ (Section~\ref{sec:conn}). \\
Inequality~\eqref{eq:reduced_constraint} can be rewritten in an extended version:
	$\beta_1u_1 + \beta_2u_2 + \dots  + \beta_nu_n \geq \epsilon$. 
	We focus on the $k$-th element:
	\begin{equation}
		\label{eq:u_k}
		\beta_k u_k \geq \epsilon - \sum_{i=1, i \neq k} ^{ n}\beta_i u _i 
	\end{equation}
	Define $\sigma_k = \sum_{i=1, i \neq k} ^{ n}\beta_i u _i$ and rewrite~\eqref{eq:u_k} as:
	\begin{equation}
		\label{eq:u_k_constraint}
		\beta_k u_k \geq \epsilon - \sigma_k
	\end{equation}
	If we consider an unbounded input, namely $U = \mathbb{R}^m$, we can always find a component $u_k$ such that:
	\begin{equation}
	\label{eq:proof2}
	u_k \geq \frac{\epsilon - \sigma}{\beta_k} \text{ if } \beta_k > 0, \quad
	u_k \leq \frac{\epsilon - \sigma}{\beta_k} \text{ if } \beta_k < 0
	\end{equation}
The definition in~\eqref{eq:proof2} is not well defined for $\beta_k = 0$. However, according to Lemma~\ref{lemma:beta_k}, it is always possible to find at least one value $k$ for which $\beta_k \neq 0$.
	
Concluding, we can always find a value $k \in [1, \dotsc, N]$ for which $u_k$ verifies~\eqref{eq:u_k_constraint} and hence~\eqref{eq:full_constraint}.  
	
\end{proof}

We will now show that the candidate CBF given in~\eqref{eq:cbf} is a valid CBF.
\begin{lemma}\label{lemma:cbf}
	Function $h(x)$ introduced in~\eqref{eq:cbf} is a valid CBF.
\end{lemma}
\begin{proof}
	In order to prove the statement, we need to show that $h(x)$ in~\eqref{eq:cbf} satisfies Properties~\ref{prop:contdiff}, \ref{prop:reldeg} and \ref{prop:existence_input} (Section~\ref{sec:cbf}).
	
	Property~\ref{prop:contdiff} requires $h(x)$ to be continuously differentiable. Considering the definition given in~\eqref{eq:cbf}, we can write
	\begin{equation*}\label{eq:dercbf}
	\frac{\partial h(x)}{\partial x} = \frac{\partial \conn}{\partial x}
	\end{equation*}
	
	As shown in~\cite{magnus1985differentiating}, simple eigenvalues of a real symmetric matrix are infinite times differentiable with respect to the variation of the element of the matrix. Considering the Laplacian matrix defined in Section~\ref{sec:conn}, such elements are the edge weights defined in~\eqref{eq:weight}, or their summation. The edge weights are by construction infinite times differentiable with respect to the state $x$ of the system. Hence, Property~\ref{prop:contdiff} is verified.

	
	Property~\ref{prop:reldeg} requires $h(x)$ to be of relative degree one. This can be easily seen computing the first order time derivative of $h(x)$ in~\eqref{eq:cbf}, as follows:
	%
	\begin{equation*}
	\frac{d}{dt}h(x)  = \frac{\partial \conn}{\partial x} \frac{d}{dt}x =  \frac{\partial \conn}{\partial x} u
	\end{equation*}
	
	Property~\ref{prop:existence_input} can be demonstrated according to Theorem~\ref{theorem:conn_constraint}.
\end{proof}

The chosen CBF~\eqref{eq:cbf} can also ensure an increase of $\conn$, till at least the threshold value $\epsilon$, if the system starts connected. This means that, if the system starts inside $\mathcal{D}$, then the input $u \in U$ that results from the QP problem will carry the system inside the set $\mathcal{C}$.
\begin{theorem}
	\label{theorem:conn_lyapunov}
	Let $\Phi_{\mathcal{C}}(x): \mathcal{D} \rightarrow \mathbb{R}$ be:
	\begin{equation}
	\label{eq:clf}
		\Phi_{\mathcal{C}}(x) = 
		\begin{cases} 
			0 & \text{if } x \in \mathcal{C}\\
			- h(x)  & \text{if } x \in \mathcal{D} \setminus \mathcal{C}  
		\end{cases}
	\end{equation}
	Consider the CBF $h(x)$ defined in~\eqref{eq:cbf}. Then, $\Phi_{\mathcal{C}}(x)$ is a CLF for the system.
\end{theorem}
\begin{proof}
	The candidate CLF~\eqref{eq:clf} is continuously differentiable and positive definite by construction. We should demonstrate that the orbital derivative $\dot{\Phi}_{\mathcal{C}}(x)$ is negative definite or semidefinite. We evaluate $\dot{\Phi}_{\mathcal{C}}(x)$:
	\begin{equation*}
		\dot{\Phi}_{\mathcal{C}}(x) = \begin{cases}
		0 & \text{if } x \in \mathcal{C}\\
		- \dot{h}(x)  & \text{if } x \in \mathcal{D} \setminus \mathcal{C} 
		\end{cases}
	\end{equation*}
	We focus on the case $x \in \mathcal{D} \setminus \mathcal{C} $:
	\begin{equation*}
	\dot{\Phi}_{\mathcal{C}}(x) = -\frac{\partial h(x)}{\partial x} \dot{x} = - \dot{h}(x) \leq \alpha\left(h(x)\right) 
	\end{equation*}
	where $\alpha(\cdot)$ is a $\mathcal{K}$ function, which is guaranteed by~\eqref{eq:sup}. \\
	Given the chosen CBF~\eqref{eq:cbf}, we know that:  
	\begin{equation*}
	\alpha\left(h(x)\right) < 0 \text{ if } x \in \mathcal{D} \setminus \mathcal{C}
	\end{equation*}
	Combining the set of inequalities above, we get:
	\begin{equation*}
	\dot{\Phi}_{\mathcal{C}}(x) < 0 
	\end{equation*}
	which ensures that the $\Phi_{\mathcal{C}}(x)$ is a CLF and asymptotically stabilizes the set $\mathcal{C}$.
%
%
%
 
\end{proof}

\section{Evaluation}
In this section we evaluate the proposed method considering a few different applications. In all the simulations, performed in Matlab, we imposed the threshold value \mbox{$\epsilon=0.1$.} 

In order to avoid collisions among the robots, thus verifying Assumption~\ref{ass:collav}, along the lines of~\cite{egerstedt2018robot}, we included an additional CBF $h_{safety}$:
\begin{equation}
\label{eq:cbf_collision}
	h_{safety}(x_i, x_j) = d_{i,j}^2 - d_{min}^2
\end{equation}
where $d_{min} > 0$ represents the distance between the robots that we consider safe, based on their the moving and sensing capabilities. This additional CBF can be considered together with $h(x)$, defined in~\eqref{eq:cbf}, in a cumulative CBF:
\begin{equation}
	h_{tot} = h \land h_{safety} = \min \{ h, h_{safety}\}
\end{equation} 
The composition of CBF is studied in~\cite{glotfelter2017nonsmooth, egerstedt2018robot} and it will not be addressed in this paper, due to space limitations. 

Extended results are shown in the attached video.

\subsection{Consensus}
The first example is based on consensus, which is a standard control strategy for multi-robot system. It is usually referred to as an aggregative behavior, but in~\cite{ji2007distributed} it has been found that, in some configurations (Fig.~\ref{fig:consensus_1}), it can carry to a disconnection of the group. Fig.~\ref{fig:consensus_lambda_2} reports the behavior of $\conn$ with and without the CBF: it is easy to see how, without the constraint on the global connectivity, the consensus in this particular configuration leads to a disconnection of the group (Fig.~\ref{fig:consensus_2}). Fig.~\ref{fig:consensus_distance} reports the distance between the robots and it demonstrates that the collision avoidance, together with the connectivity maintenance, can be achieved with the proposed method. The final configuration is reported in Fig.~\ref{fig:consensus_3}.
\begin{figure}[tb]
	\renewcommand\thesubfigure{(\alph{subfigure})}
	\centering
	\begin{subfigure}{0.47\columnwidth}
		\fbox{\includegraphics[width=\columnwidth]{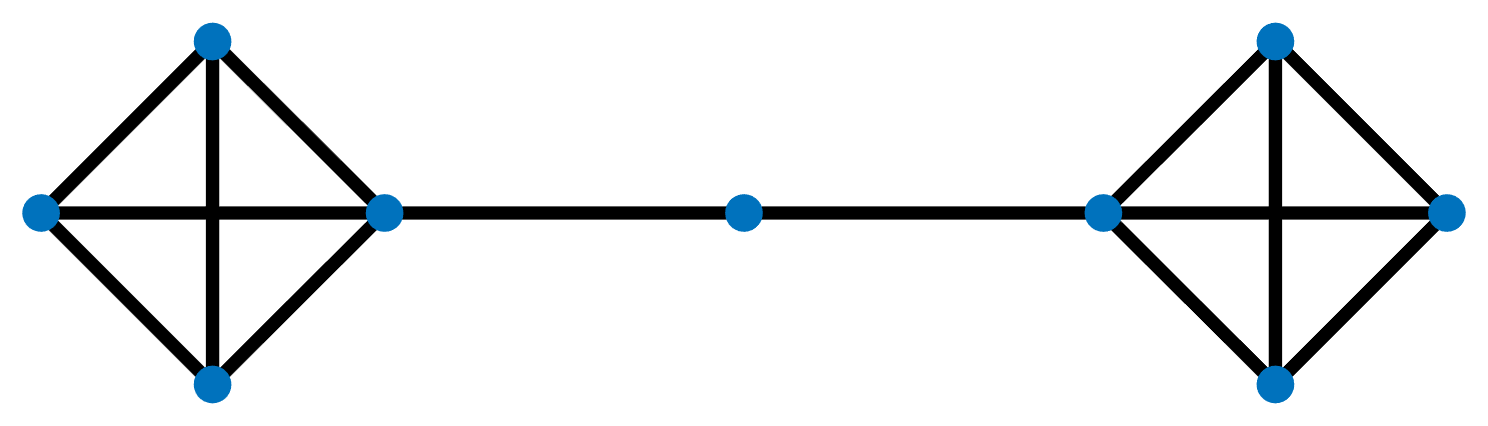}}
		\subcaption{$t = 0 \si{\second}$.}
		\label{fig:consensus_1}
	\end{subfigure}\\
	\begin{subfigure}{0.47\columnwidth}
		\fbox{\includegraphics[width=\columnwidth]{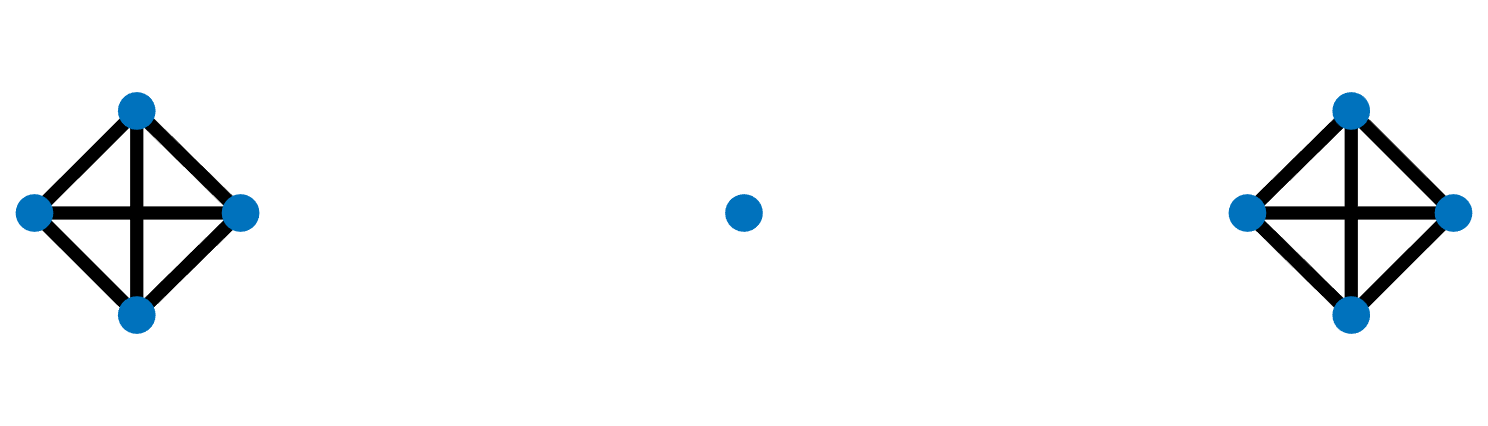}}
		\subcaption{$t = 0.5 \si{\second}$, $h(x)$ disabled.}
		\label{fig:consensus_2}
	\end{subfigure}\hspace{0.2cm}
	\begin{subfigure}{0.47\columnwidth}
		\fbox{\includegraphics[width=\columnwidth]{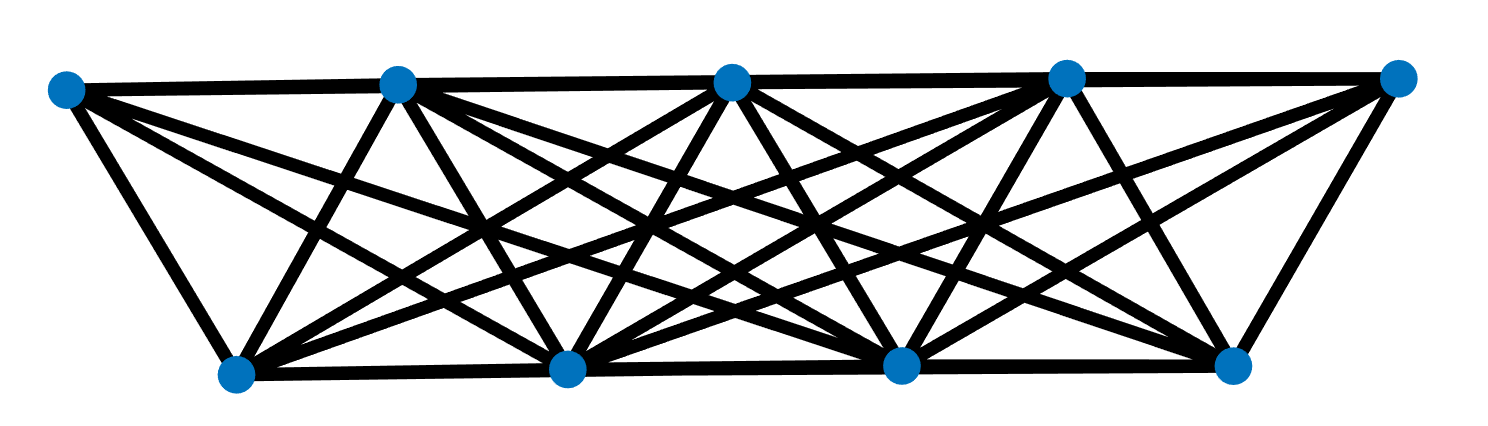}}
		\subcaption{$t = 20 \si{\second}$, $h(x)$ enabled.}
		\label{fig:consensus_3}
	\end{subfigure}
	\caption{Some snapshots of the consensus simulation.}
	\label{fig:snapshot}
\end{figure}
\begin{figure}[tb]
	\renewcommand\thesubfigure{(\alph{subfigure})}
	\centering
	\begin{subfigure}{0.52\columnwidth}
		\begin{tikzpicture}
		\begin{axis}[
		height = 4.6cm,
		width = 4.6cm,
		xlabel near ticks, 
		xlabel = {t [\si{\second}]},
		ylabel near ticks,  
		ylabel shift={-4pt},
		ylabel = {$\conn$ },
		ymin=-0.1, ymax=1.05,
		ytick={0, 0.5, 1},
		yticklabels = {$0$, $0.5$, $1$},
		xmin=0, xmax=20,
		legend pos=north west,
		legend image post style={only marks, mark=-},
		legend style={font=\footnotesize}
		]
		\addplot [line width=1.0pt, color=red] 
		table[x = x_pos, y = y_pos]
		{
			x_pos	y_pos
			0  0.1
			20  0.1
		};
		\addplot [line width=1.4pt, color=color1] table [x index=2,y index=0,col sep=comma] {data/consensus_lambda_2_all.csv};
		\addplot [line width=1.4pt, color=color6] table [x index=2,y index=1,col sep=comma] {data/consensus_lambda_2_all.csv};
		\legend{$\epsilon = 0.1$}
		
		\end{axis}
		\end{tikzpicture}
		\subcaption{Value of $\conn$ for: \color{color1}\ding{110} \color{black}$h(x)$ \\ enabled, \color{color6}\ding{110} \color{black}$h(x)$ disabled.}
		\label{fig:consensus_lambda_2}
	\end{subfigure}
	\begin{subfigure}{0.52\columnwidth}
		\begin{tikzpicture}
		\begin{axis}[
		height = 4.6cm,
		width = 4.6cm,
		xlabel near ticks, 
		xlabel = {t [\si{\second}]},
		ylabel near ticks, 
		ylabel shift={-4pt},
		ylabel = {$d_{i,j}$ [\si{\m}]},
		ymin=0, ymax=16,
		ytick={0, 5, 10, 15},
		yticklabels = { $0$, $5$, $10$, $15$},
		xmin=0, xmax=20,
		cycle list name=MyCyclelist	,
		legend image post style={only marks, mark=-},
		legend style={font=\footnotesize}
		]
		\addplot [line width=1pt, color=red] 
		table[x = x_pos, y = y_pos]
		{
			x_pos	y_pos
			0  1.5
			20  1.5
		};
		\foreach \i in {0,1,...,71} {\addplot+  table [x index=72,y index=\i,col sep=comma] {data/consensus_distance.csv}; }%
		%
		\legend{$d_{min} = 1.5$}
		\end{axis}
		\end{tikzpicture}
		\subcaption{Distance between the neighboring robots.}
		\label{fig:consensus_distance}
	\end{subfigure}
\caption{Graphs of the consensus simulation.}
\label{fig:consensus}
\end{figure}
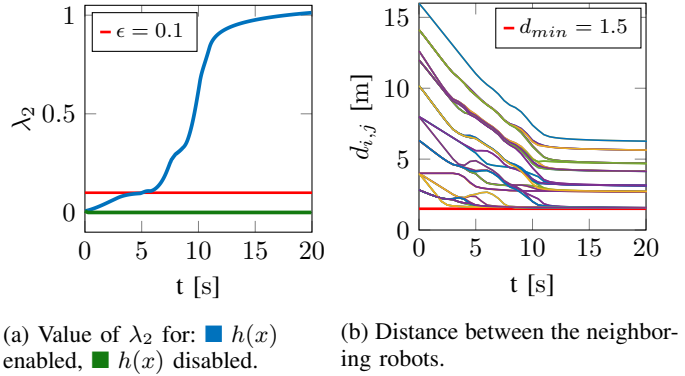

\subsection{Disconnecting control action}
In this section we want to demonstrate that the proposed method can maintain connectivity against a desired controller that tends to break down the group. The desired controller is defined for the $i$-th robot, where $i \in [1, \dotsc, N]$, as:
\begin{equation}
\label{eq:disgregative}
	u_{des}^i = \left[ 
	k \cos \left( \frac{2 \pi}{N + 1} i\right) \quad
	k \sin \left( \frac{2 \pi}{N + 1} i\right)
	 \right]^T
\end{equation}
where $k > 0$ is a tuning parameter. 

We report two main trials of this series of simulations. The first has 4 robots and the system starts with $0 < \conn < \epsilon$, namely $x(0) \in \mathcal{D} \setminus \mathcal{C}$. Fig~\ref{fig:disgregative_4_lambda_2} reports the trend of $\conn$ during the simulation and, as stated in Theorem~\ref{theorem:conn_lyapunov}, the CLF defined as~\eqref{eq:clf} makes the set $\mathcal{C}$ asymptotically stable and the chosen CBF can actually increase the connectivity value up to the desired threshold.

\begin{figure}[tb]
	\renewcommand\thesubfigure{(\alph{subfigure})}
	\centering
	\begin{subfigure}{0.52\columnwidth}
		\begin{tikzpicture}
		\begin{axis}[
		height = 4.6cm,
		width = 4.6cm,
		xlabel near ticks, 
		xlabel = {t [\si{\second}]},
		ylabel near ticks,  
		ylabel shift={-4pt},
		ylabel = {$\conn$ },
		ymin=-0.05, ymax=0.35,
			ytick={0, 0.1, 0.2, 0.3},
			yticklabels = {$0$, $0.1$, $0.2$, $0.3$},
		xmin=0, xmax=20,
		legend image post style={only marks, mark=-},
		legend style={font=\footnotesize}
		]
		
		\addplot [line width=1.0pt, color=red] 
		table[x = x_pos, y = y_pos]
		{
			x_pos	y_pos
			0  0.1
			20  0.1
		};
		
		\addplot [line width=1.4pt, color=color1] table [x index=1,y index=0,col sep=comma] {data/disgregative_4_lambda_2.csv};
		
		\legend{$\epsilon = 0.1$}
		\end{axis}
		\end{tikzpicture}
		\subcaption{$N = 4$, $\conn(x(0)) < \epsilon$. }
		\label{fig:disgregative_4_lambda_2}
	\end{subfigure}
	\begin{subfigure}{0.52\columnwidth}
		\begin{tikzpicture}
		\begin{axis}[
		height = 4.6cm,
		width = 4.6cm,
		xlabel near ticks, 
		xlabel = {t [\si{\second}]},
		ylabel near ticks,  
		ylabel shift={-4pt},
		ylabel = {$\conn$},
		ymin=-0.5, ymax=4.5,
		ytick={0, 2, 4},
		yticklabels = {$0$, $2$, $4$},
		xmin=0, xmax=20,
		legend image post style={only marks, mark=-},
		legend style={font=\footnotesize}	
		]
		
		\addplot [line width=1.0pt, color=red] 
		table[x = x_pos, y = y_pos]
		{
			x_pos	y_pos
			0  0.1
			20  0.1
		};
	
		\addplot [line width=1.4pt, color=color1] table [x index=1,y index=0,col sep=comma] {data/disgregative_40_lambda_2.csv};
		\legend{$\epsilon=0.1$}
		\end{axis}
		\end{tikzpicture}
		\subcaption{$N = 40$, $\conn(x(0)) > \epsilon$.}
		\label{fig:disgregative_40}
	\end{subfigure}
	\caption{Value of $\conn$ in simulations with disconnecting $u_{des}$~\eqref{eq:disgregative}.}
	\label{fig:disgregative}
\end{figure}
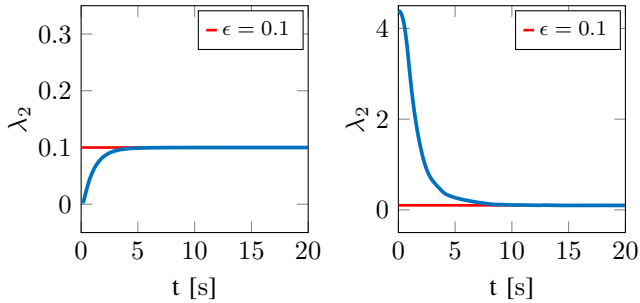

In the second example, the system, composed of 40 robots, starts inside the set $\mathcal{C}$. Fig.~\ref{fig:disgregative_40} reports the behavior of $\conn$.

\subsection{Coverage}
In many typical applications of multi-robot system, the main objective is searching or monitoring a certain quantity or data with ad-hoc sensors. Usually the area $\Upsilon \in \mathbb{R}^n$ in which the robots are deployed has some points of interest, e.g., a higher probability to find victims in a search and rescue scenario or to measure the quantity that should be monitored. This different importance of the zones is usually described by means of a density function associated to the domain~\cite{cortes2004coverage, cortes2010coverage}, namely $\phi(\cdot) : \Upsilon \rightarrow \mathbb{R}_0^+$. These density functions can be defined for robot with different sensors~\cite{santos2018coverage}, exploiting heterogeneity in multi-robot systems, or time-varying~\cite{lee2015multirobot}, to consider the evolution of the area during the task. However, spreading robots over an area can lead to loss of connectivity.


In this simulation, we solved the QP problem~\eqref{eq:qp} using, as desired input, the one coming from a Voronoi-based coverage algorithm. In particular, we chose to implement the standard Lloyd approach~\cite{lloyd1982least, cortes2004coverage}, which defines the input for the $i$-th robot as 
$
u_{des}^i = -k \left( x_i - c_i \right)
$, 
where $k > 0$ is a tuning parameter, 
and $c_i$ is the center of mass of the $i$-th Voronoi cell $\mathcal{V}_i$:
\begin{equation*}
\mathcal{V}_i(x) = \{q \in \Upsilon | \|q- x_i \| \leq \| q - x_j \|, \forall j \in \left[1, \dotsc, N \right] \}
\end{equation*}

In addition, we report a comparison with a local connectivity approach, used in~\cite{egerstedt2018robot, wenhao2019}. The approach proposed in these works is based on finding a CBF $h_{local}$ that defines the desired set as the set that maintains all the initial links:
\begin{equation}
 h_{local}(x_i, x_j) = R^2 - d_{i,j}^2 
\end{equation} 
This CBF, as $h_{safety}$~\eqref{eq:cbf_collision}, is defined for each pair of connected robots and not for the whole system as, vice-versa, our $h(x)$~\eqref{eq:cbf}.
The comparison is based on the locational cost~\cite{cortes2004coverage}, which is usually used to measure the efficiency of the coverage algorithms (lower value indicates better performance):
\begin{equation}
\label{eq:locational_cost}
\mathcal{H}(x) = \sum_{i = 1}^{N} \int_{\mathcal{V}_i(x)} \|q - x_i \|^2\phi(q)dq
\end{equation}

We report, in the following, a significant example that
highlights the differences between the proposed method and the existing one. The simulation was conducted with 4 robots, starting from a connected configuration above the given threshold ($x(0) \in \mathcal{C}$). Fig.~\ref{fig:coverage_lambda_2} reports the behavior of the algebraic connectivity and, as might be expected, the proposed method keeps the value above the threshold, while the trial with $h_{local}$ keeps the group connected, but with a lower value of $\conn$. In addition, the proposed method allows to obtain a lower value of locational cost~\eqref{eq:locational_cost}, which is compared to the one obtained with $h_{local}$ in Fig.~\ref{fig:coverage_HC}. 

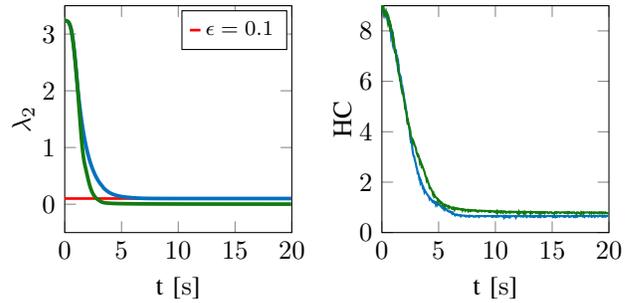
\begin{figure}[tb]
	\renewcommand\thesubfigure{(\alph{subfigure})}
	\centering
	\begin{subfigure}{0.48\columnwidth}
		\begin{tikzpicture}
		\begin{axis}[
		height = 4.6cm,
		width = 4.6cm,
		xlabel near ticks, 
		xlabel = {t [\si{\second}]},
		ylabel near ticks,  
		ylabel shift={-4pt}, 
		ylabel = {$\conn$},
		ymin=-0.5, ymax=3.5,
		ytick={ 0, 1, 2, 3},
		yticklabels = { $0$, $1$, $2$, $3$},
		xmin=0, xmax=20,
		legend image post style={only marks, mark=-},
		legend style={font=\footnotesize}	
		]
		
		\addplot [line width=1pt, color=red 	] 
		table[x = x_pos, y = y_pos]
		{
			x_pos	y_pos
			0  0.1
			20  0.1
		};
	
		\addplot [line width=1.4pt, color=color1] table [x index=4,y index=0,col sep=comma] {data/coverage_all.csv};
		\addplot [line width=1.4pt, color=color6] table [x index=4,y index=2,col sep=comma] {data/coverage_all.csv};
		\legend{$\epsilon=0.1$}
		\end{axis}
		\end{tikzpicture}
		\subcaption{Value of $\conn$.}
		\label{fig:coverage_lambda_2}
	\end{subfigure}
	\begin{subfigure}{0.48\columnwidth}
		\begin{tikzpicture}
		\begin{axis}[
		height = 4.6cm,
		width = 4.6cm,
		xlabel near ticks, 
		xlabel = {t [\si{\second}]},
		ylabel near ticks,  
		ylabel shift={-4pt},
		ylabel = {HC},
		ymin=0, ymax=9,
		xmin=0, xmax=20,	
		]
		\addplot [line width=0.8pt, color=color1] table [x index=4,y index=1,col sep=comma] {data/coverage_all.csv};
		\addplot [line width=0.8pt, color=color6] table [x index=4,y index=3,col sep=comma] {data/coverage_all.csv};
		
		\end{axis}
		\end{tikzpicture}
		\subcaption{Value of $H(x)$.}
		\label{fig:coverage_HC}
	\end{subfigure}
	\caption{Coverage simulation with comparison between: \\ \color{color1}\ding{110} \color{black}$h(x)$ and \color{color6}\ding{110} \color{black}$h_{local}$.}
	\label{fig:coverage}
\end{figure}

\section{Conclusion}
In this paper we propose a method to incorporate, in an optimal way, connectivity maintenance and a desired controller. The method is based on the definition of a Control Barrier Function that ensures the maintenance of the algebraic connectivity above a given value. The simulations show the feasibility and the flexibility of the method over different typical tasks. It is important to highlight that the proposed method was developed considering a centralized implementation. About the future work, we are going to implement the proposed method in a decentralized fashion. Starting from the works presented in~\cite{borrmann2015control, wang2017safety}, where a decentralized version of the CBF is introduced, and building upon the results in~\cite{sabattini2013decentralized} for decentralized estimation of the algebraic connectivity.
  
%
%

\bibliographystyle{IEEEtran}
\bibliography{biblio}

\begin{thebibliography}{10}
\providecommand{\url}[1]{#1}
\csname url@rmstyle\endcsname
\providecommand{\newblock}{\relax}
\providecommand{\bibinfo}[2]{#2}
\providecommand\BIBentrySTDinterwordspacing{\spaceskip=0pt\relax}
\providecommand\BIBentryALTinterwordstretchfactor{4}
\providecommand\BIBentryALTinterwordspacing{\spaceskip=\fontdimen2\font plus
\BIBentryALTinterwordstretchfactor\fontdimen3\font minus
  \fontdimen4\font\relax}
\providecommand\BIBforeignlanguage[2]{{%
\expandafter\ifx\csname l@#1\endcsname\relax
\typeout{** WARNING: IEEEtran.bst: No hyphenation pattern has been}%
\typeout{** loaded for the language `#1'. Using the pattern for}%
\typeout{** the default language instead.}%
\else
\language=\csname l@#1\endcsname
\fi
#2}}

\bibitem{ram2017}
L.~Sabattini, M.~Aikio, P.~Beinschob, M.~Boehning, E.~Cardarelli, V.~Digani,
  A.~Krengel, M.~Magnani, S.~Mandici, F.~Oleari, C.~Reinke, D.~Ronzoni,
  C.~Stimming, R.~Varga, A.~Vatavu, S.~Castells~Lopez, C.~Fantuzzi,
  A.~{M\"ayr\"a}, S.~Nedevschi, C.~Secchi, and K.~Fuerstenberg, ``The
  pan-robots project: Advanced automated guided vehicle systems for industrial
  logistics,'' \emph{IEEE Robotics Automation Magazine}, vol.~25, no.~1, pp.
  55--64, March 2018.

\bibitem{ball2015}
D.~Ball, P.~Ross, A.~English, T.~Patten, B.~Upcroft, R.~Fitch, S.~Sukkarieh,
  G.~Wyeth, and P.~Corke, ``Robotics for sustainable broad-acre agriculture,''
  in \emph{Field and Service Robotics}.\hskip 1em plus 0.5em minus 0.4em\relax
  Springer, 2015, pp. 439--453.

\bibitem{hajieghrary2016}
H.~Hajieghrary, M.~A. Hsieh, and I.~B. Schwartz, ``Multi-agent search for
  source localization in a turbulent medium,'' \emph{Physics Letters A}, vol.
  380, no.~20, pp. 1698--1705, 2016.

\bibitem{kemna2018}
S.~Kemna, H.~Heiarsson, and G.~S. Sukhatme, ``On-board adaptive informative
  sampling for auvs: a feasibility study,'' in \emph{Proc. 2018 MTS/IEEE
  OCEANS}.\hskip 1em plus 0.5em minus 0.4em\relax IEEE, 2018, pp. 1--10.

\bibitem{saska2017}
M.~Saska, T.~Baca, J.~Thomas, J.~Chudoba, L.~Preucil, T.~Krajnik, J.~Faigl,
  G.~Loianno, and V.~Kumar, ``System for deployment of groups of unmanned micro
  aerial vehicles in gps-denied environments using onboard visual relative
  localization,'' \emph{Autonomous Robots}, vol.~41, no.~4, pp. 919--944, 2017.

\bibitem{hayat2017multi}
S.~Hayat, E.~Yanmaz, T.~X. Brown, and C.~Bettstetter, ``Multi-objective uav
  path planning for search and rescue,'' in \emph{Proc. 2017 IEEE Int. Conf. on
  Robotics and Automation (ICRA)}.\hskip 1em plus 0.5em minus 0.4em\relax IEEE,
  2017, pp. 5569--5574.

\bibitem{portugal2011survey}
D.~Portugal and R.~Rocha, ``A survey on multi-robot patrolling algorithms,'' in
  \emph{Doctoral Conf. on Computing, Electrical and Industrial Systems}.\hskip
  1em plus 0.5em minus 0.4em\relax Springer, 2011, pp. 139--146.

\bibitem{pasqualetti2012cooperative}
F.~Pasqualetti, A.~Franchi, and F.~Bullo, ``On cooperative patrolling: Optimal
  trajectories, complexity analysis, and approximation algorithms,'' \emph{IEEE
  Trans. on Robotics}, vol.~28, no.~3, pp. 592--606, 2012.

\bibitem{cortes2010coverage}
J.~Cort{\'e}s, ``Coverage optimization and spatial load balancing by robotic
  sensor networks,'' \emph{IEEE Trans. on Automatic Control}, vol.~55, no.~3,
  pp. 749--754, 2010.

\bibitem{burgard2005coordinated}
W.~Burgard, M.~Moors, C.~Stachniss, and F.~E. Schneider, ``Coordinated
  multi-robot exploration,'' \emph{IEEE Trans. on Robotics}, vol.~21, no.~3,
  pp. 376--386, 2005.

\bibitem{ji2007distributed}
M.~Ji and M.~Egerstedt, ``Distributed coordination control of multiagent
  systems while preserving connectedness,'' \emph{IEEE Trans. on Robotics},
  vol.~23, no.~4, pp. 693--703, 2007.

\bibitem{dimarogonas2008decentralized}
D.~V. Dimarogonas and K.~H. Johansson, ``Decentralized connectivity maintenance
  in mobile networks with bounded inputs,'' in \emph{Proc. 2008 IEEE Int. Conf.
  on Robotics and Automation (ICRA)}.\hskip 1em plus 0.5em minus 0.4em\relax
  IEEE, 2008, pp. 1507--1512.

\bibitem{ajorlou2010class}
A.~Ajorlou, A.~Momeni, and A.~G. Aghdam, ``A class of bounded distributed
  control strategies for connectivity preservation in multi-agent systems,''
  \emph{IEEE Trans. on Automatic Control}, vol.~55, no.~12, pp. 2828--2833,
  2010.

\bibitem{yang2010decentralized}
P.~Yang, R.~A. Freeman, G.~J. Gordon, K.~M. Lynch, S.~S. Srinivasa, and
  R.~Sukthankar, ``Decentralized estimation and control of graph connectivity
  for mobile sensor networks,'' \emph{Automatica}, vol.~46, no.~2, pp.
  390--396, 2010.

\bibitem{sabattini2013decentralized}
L.~Sabattini, N.~Chopra, and C.~Secchi, ``Decentralized connectivity
  maintenance for cooperative control of mobile robotic systems,'' \emph{The
  Int. J. of Robotics Research}, vol.~32, no.~12, pp. 1411--1423, 2013.

\bibitem{li2013bounded}
X.~Li, D.~Sun, and J.~Yang, ``A bounded controller for multirobot navigation
  while maintaining network connectivity in the presence of obstacles,''
  \emph{Automatica}, vol.~49, no.~1, pp. 285--292, 2013.

\bibitem{khateri2019comparison}
K.~Khateri, M.~Pourgholi, M.~Montazeri, and L.~Sabattini, ``A comparison
  between decentralized local and global methods for connectivity maintenance
  of multi-robot networks,'' \emph{IEEE Robotics and Automation Letters},
  vol.~4, no.~2, pp. 633--640, 2019.

\bibitem{ames2014control}
A.~D. Ames, J.~W. Grizzle, and P.~Tabuada, ``Control barrier function based
  quadratic programs with application to adaptive cruise control,'' in
  \emph{Proc. 53rd IEEE Conf. on Decision and Control}.\hskip 1em plus 0.5em
  minus 0.4em\relax IEEE, 2014, pp. 6271--6278.

\bibitem{xu2017realizing}
X.~Xu, T.~Waters, D.~Pickem, P.~Glotfelter, M.~Egerstedt, P.~Tabuada, J.~W.
  Grizzle, and A.~D. Ames, ``Realizing simultaneous lane keeping and adaptive
  speed regulation on accessible mobile robot testbeds,'' in \emph{Proc. 2017
  IEEE Conf. on Control Technology and Applications (CCTA)}.\hskip 1em plus
  0.5em minus 0.4em\relax IEEE, 2017, pp. 1769--1775.

\bibitem{landi2019safety}
C.~T. Landi, F.~Ferraguti, S.~Costi, M.~Bonf{\`e}, and C.~Secchi, ``Safety
  barrier functions for human-robot interaction with industrial manipulators,''
  in \emph{Proc. 18th European Control Conf. (ECC)}.\hskip 1em plus 0.5em minus
  0.4em\relax IEEE, 2019, pp. 2565--2570.

\bibitem{wang2017safety}
L.~Wang, A.~D. Ames, and M.~Egerstedt, ``Safety barrier certificates for
  collisions-free multirobot systems,'' \emph{IEEE Trans. on Robotics},
  vol.~33, no.~3, pp. 661--674, 2017.

\bibitem{egerstedt2018robot}
M.~Egerstedt, J.~N. Pauli, G.~Notomista, and S.~Hutchinson, ``Robot ecology:
  Constraint-based control design for long duration autonomy,'' \emph{Annual
  Reviews in Control}, 2018.

\bibitem{wenhao2019}
W.~Luo and K.~Sycara, ``Voronoi-based coverage control with connectivity
  maintenance for robotic sensor networks,'' in \emph{Proc. 2019 Int. Symp. on
  Multi-Robot and Multi-Agent Systems (MRS)}.\hskip 1em plus 0.5em minus
  0.4em\relax IEEE, 2019.

\bibitem{siligardi2019robust}
L.~Siligardi, J.~Panerati, M.~Kaufmann, M.~Minelli, C.~Ghedini, G.~Beltrame,
  and L.~Sabattini, ``Robust area coverage with connectivity maintenance,'' in
  \emph{Proc. 2019 Int. Conf. on Robotics and Automation (ICRA)}.\hskip 1em
  plus 0.5em minus 0.4em\relax IEEE, 2019, pp. 2202--2208.

\bibitem{mohar1991laplacian}
B.~Mohar, Y.~Alavi, G.~Chartrand, and O.~Oellermann, ``The laplacian spectrum
  of graphs,'' \emph{Graph theory, combinatorics, and applications}, vol.~2,
  no. 871-898, p.~12, 1991.

\bibitem{ames2016control}
A.~D. Ames, X.~Xu, J.~W. Grizzle, and P.~Tabuada, ``Control barrier function
  based quadratic programs for safety critical systems,'' \emph{IEEE Trans. on
  Automatic Control}, vol.~62, no.~8, pp. 3861--3876, 2016.

\bibitem{khalil2002nonlinear}
H.~K. Khalil, ``Nonlinear systems,'' \emph{Upper Saddle River}, 2002.

\bibitem{tro2017}
A.~Gasparri, L.~Sabattini, and G.~Ulivi, ``Bounded control law for global
  connectivity maintenance in cooperative multi-robot systems,'' \emph{IEEE
  Trans. on Robotics}, vol.~33, no.~3, pp. 700--717, June 2017.

\bibitem{poignard2018spectra}
C.~Poignard, T.~Pereira, and J.~P. Pade, ``Spectra of laplacian matrices of
  weighted graphs: structural genericity properties,'' \emph{SIAM J. on Applied
  Mathematics}, vol.~78, no.~1, pp. 372--394, 2018.

\bibitem{magnus1985differentiating}
J.~R. Magnus, ``On differentiating eigenvalues and eigenvectors,''
  \emph{Econometric Theory}, vol.~1, no.~2, pp. 179--191, 1985.

\bibitem{glotfelter2017nonsmooth}
P.~Glotfelter, J.~Cort{\'e}s, and M.~Egerstedt, ``Nonsmooth barrier functions
  with applications to multi-robot systems,'' \emph{IEEE control systems
  letters}, vol.~1, no.~2, pp. 310--315, 2017.

\bibitem{cortes2004coverage}
J.~Cortes, S.~Martinez, T.~Karatas, and F.~Bullo, ``Coverage control for mobile
  sensing networks,'' \emph{IEEE Trans. on Robotics and Automation}, vol.~20,
  no.~2, pp. 243--255, 2004.

\bibitem{santos2018coverage}
M.~Santos, Y.~Diaz-Mercado, and M.~Egerstedt, ``Coverage control for multirobot
  teams with heterogeneous sensing capabilities,'' \emph{IEEE Robotics and
  Automation Letters}, vol.~3, no.~2, pp. 919--925, 2018.

\bibitem{lee2015multirobot}
S.~G. Lee, Y.~Diaz-Mercado, and M.~Egerstedt, ``Multirobot control using
  time-varying density functions,'' \emph{IEEE Trans. on Robotics}, vol.~31,
  no.~2, pp. 489--493, 2015.

\bibitem{lloyd1982least}
S.~Lloyd, ``Least squares quantization in pcm,'' \emph{IEEE Trans. on Inform.
  Theory}, vol.~28, no.~2, pp. 129--137, 1982.

\bibitem{borrmann2015control}
U.~Borrmann, L.~Wang, A.~D. Ames, and M.~Egerstedt, ``Control barrier
  certificates for safe swarm behavior,'' \emph{IFAC-PapersOnLine}, vol.~48,
  no.~27, pp. 68--73, 2015.

\end{thebibliography}

\end{document}